\documentclass[letterpaper, 10 pt, conference]{ieeeconf}  % Comment this line out
\IEEEoverridecommandlockouts                              % This command is only
\let\labelindent\relax
\usepackage{graphicx}
\usepackage[bottom]{footmisc}
\usepackage{amsmath,amssymb, amsthm}
\usepackage{booktabs}
\usepackage{enumitem}
\usepackage{thmtools}
\usepackage{amsfonts,bm}
\newcommand{\halfspace}{\kern 0.1em}
\usepackage{dsfont}
\usepackage{color}
\usepackage{xcolor}
\usepackage[compress]{cite}

\usepackage{wrapfig}

\newcommand{\V}{\mathcal{V}}

\newcommand{\T}{\mathcal{T}}
\renewcommand{\S}{\mathcal{S}}
\newcommand{\U}{\mathcal{U}}

\newcommand{\dAC}{DR}

\newcommand{\optism}{\mathrm{optsm}}
\newcommand{\pess}{\mathrm{pess}}
\newcommand{\true}{\mathrm{true}}

\newcommand{\pursuer}{Pursuer}
\newcommand{\evader}{Evader}
\newcommand{\Preqx}{\mathcal{P}_{\mathrm{req}}}

\newcommand{\Rsetx}{\mathcal{R}_x}
\newcommand{\Rcap}{\mathcal{R}_\mathrm{cap}}
\newcommand{\Rsety}{\mathcal{R}_y}
\newcommand{\Rinvsetx}{\mathcal{R}^{\scalebox{0.7}{$-1$}}_x}

\newcommand{\Qsetx}{\mathcal{Q}}
\newcommand{\Zsetx}{\mathcal{Z}}
\newcommand{\Qsetinf}{\mathcal{Q}_\infty}

\newcommand{\capture}{\mathrm{Capt}}
\newcommand{\invcapture}{\capture^{\scalebox{0.7}{$-1$}}}
\newcommand{\target}{\mathcal{T}}
\newcommand{\Qplus}{\mathcal{Q}^+}
\newcommand{\Gset}{\mathcal{G}}

\newcommand{\midcap}[1]{
\mathop{\scalebox{1.8}{$\cap$}}\limits_{\scriptscriptstyle#1} 
}
\newcommand{\midcup}[1]{
\mathop{\scalebox{1.8}{$\cup$}}\limits_{\scriptscriptstyle#1} 
}

\renewcommand{\emptyset}{\varnothing}

\newtheorem{theorem}{Theorem}
\newtheorem{corollary}[theorem]{Corollary}
\newtheorem{lemma}[theorem]{Lemma}
\newtheorem{proposition}[theorem]{Proposition}
\newtheorem{definition}{Definition}

\newtheorem{remark}{Remark}

\usepackage{algorithmic}
\usepackage[ruled,vlined,linesnumbered]{algorithm2e}    % Algorithm line numbered
\newcommand{\nonl}{\renewcommand{\nl}{\let\nl\oldnl}}   % Remove line number for one line
\usepackage{listings}

\usepackage{environ}

\NewEnviron{resizealign}{\sbox0{
    $\begin{matrix}\displaystyle\BODY\end{matrix}$}%
  \sbox1{$(\theequation)$}%
  \sbox2{\parbox{\dimexpr \wd0 + 2\wd1}%
    {\begin{align*}\BODY\end{align*}}}% for testing
  \noindent\resizebox{\textwidth}{!}{\usebox2}%
}

\title{\LARGE \bf
Feedback Dominance Analysis for Pursuit–Evasion Games on Graphs
}

\author{
    Yue Guan$^{1}$ 
    ~~ Daigo Shishika$^{2}$ 
    ~~ Dipankar Maity$^{3}$ 
    ~~ Michael Dorothy$^{4}$
    ~~ Panagiotis Tsiotras$^{1}$
    \thanks{The views expressed in this paper are those of the authors and do not reflect the official policy or position of the United States Army, Department of Defense, or the United States Government.}
    \thanks{$^{1}$Yue Guan and Panagiotis Tsiotras are with the School of Aerospace Engineering, Georgia Institute of Technology, Atlanta, GA, USA.
    {\tt\small $\{$yguan44, tsiotras$\}$@gatech.edu}}
    \thanks{$^{2}$Daigo Shishika is with the College of Engineering and Computing, George Mason University, Fairfax, VA, USA.
    {\tt\small dshishik@gmu.edu}}
    \thanks{$^{3}$Dipankar Maity is with the Department of Electrical and Computer Engineering, University of North Carolina at Charlotte, Charlotte, NC, USA.
    {\tt\small dmaity@charlotte.edu}}
    \thanks{$^{4}$Michael Dorothy is with DEVCOM Army Research Laboratory, Adelphi, MD, USA. {\tt\small michael.r.dorothy.civ@army.mil}}
}

\begin{document}

\maketitle
\thispagestyle{empty}
\pagestyle{empty}

%%%%%%%%%%%%%%%%%%%%%%%%%%%%%%%%%%%%%%%%%%%%%%%%%%%%%%%%%%%%%%%%%%%%%%%%%%%%%%%%
\begin{abstract}
    This work identifies the dominance regions for discrete, simultaneous-move pursuit-evasion games on graphs. 
    Existing geometric approaches provide efficient characterizations of winning regions, but typically provide only sufficient conditions and rely on open-loop strategies. 
    To address these challenges, we develop a set-based dynamic programming approach to characterize the \pursuer{}’s winning and losing regions, providing necessary and sufficient winning conditions under worst-case behavior. 
    The reachability analysis admits a set-chasing interpretation, allowing translation of dominance sets to feedback strategies that adapt to the players’ positions in real time. 
    For states where neither player can guarantee victory, we introduce an instantaneous matrix-game formulation and establish upper and lower bounds on the \pursuer{}’s winning probability. 
    Simulation results validate the correctness of the dominance-region characterization and the proposed bounds.
\end{abstract}

%%%%%%%%%%%%%%%%%%%%%%%%%%%%%%%%%%%%%%%%%%%%%%%%%%%%%%%%%%%%%%%%%%%%%%%%%%%%%%%%
\section{Introduction}

Pursuit--evasion games (PEGs) have been studied extensively due to their fundamental role in defense~\cite{isaacs1999differential}, security~\cite{chung2011search}, and multi-robot systems~\cite{shishika2020review}. 
In its basic form, a PEG, also known as the \emph{cops-and-robbers} game~\cite{nowakowski1983vertex}, considers an \evader{} attempting to avoid capture by a \pursuer{}. 
Many important variants extend this formulation, including \emph{reach-avoid} games~\cite{chen2016reach} with added evasion states, \emph{perimeter-defense} games~\cite{shishika2020review} with \pursuer{}s restricted to a region's perimeter, \emph{target-defense}~\cite{von2022circular} games with turrets aiming at \evader{}s, etc. 

A standard approach to solving PEGs is through a game-of-degree formulation, where agents optimize performance metrics such as time-to-capture or time-to-escape. 
In continuous domains, such problems can be formulated as differential games and solved via Hamilton--Jacobi--Isaacs (HJI) equations~\cite{basar1999dynamic}, yielding optimal \emph{feedback} strategies. 
In discrete settings, PEGs can be modeled as stochastic games, where Nash equilibria can be computed using value iteration~\cite{owen2013game}. 
While these approaches are general and provide guarantees, they do not exploit the underlying geometric structure and are therefore computationally expensive, if not intractable.

An alternative line of work leverages geometric structure to obtain efficient solutions for PEGs. 
In continuous obstacle-free environments, reachability-based constructions such as the Apollonius circle (AC)~\cite{ramana2017pursuit, pachter2019two} characterize the set of points that each agent can reach first. 
This leads to the notion of \emph{dominance regions}, where, roughly speaking, if the \evader{} can reach an evasion state strictly earlier than the \pursuer{}, then successful evasion is guaranteed. 
Such approaches are computationally efficient and admit closed-form solutions~\cite{dorothy2024one}. 
However, they provide only \emph{sufficient} conditions and typically yield \emph{open-loop strategies} based on shortest-path comparisons.

This limitation becomes critical in environments with obstacles. 
Even when the \pursuer{} can reach all evasion states faster than the \evader{}, it may still fail to guarantee capture due to \emph{dilemmas}, where committing to protect one evasion state allows the \evader{} to escape through others. 
In such cases, geometric shortest-path comparisons become overly conservative and fail to fully capture the strategic interaction.

\begin{figure}[t]
    \centering
    \includegraphics[width=0.35\linewidth]{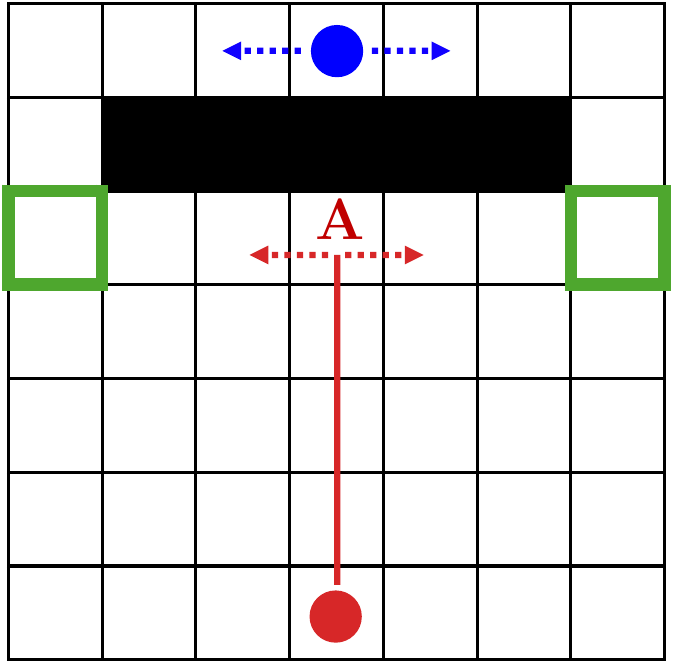}
    \vspace{-0.1in}
    \caption{Dilemma faced by the (blue) \pursuer{}, exploited by the (red) \evader{}. 
    Black cells denote obstacles, and green squares mark evasion states.}
    \label{fig:dilemma}
    \vspace{-0.25in}
\end{figure}

To illustrate this limitation, consider the scenario in Fig.~\ref{fig:dilemma}. 
Although the \pursuer{} is strictly closer to both evasion states, it must commit immediately to a specific direction due to the obstacle, while the \evader{} can delay its decision until it reaches point~$A$. 
As a result, the \evader{} can employ a feedback strategy that guarantees successful evasion by conditioning on the \pursuer{}’s position when the \evader{} reaches point~$A$.
This behavior cannot be captured by existing open-loop geometric analyses, thus motivating the feedback-based dominance analysis proposed in this work.

To address these challenges, prior works have introduced additional structural assumptions, such as sequential actions or restricted motion models~\cite{amigoni2012game, fomin2011guard, nagamochi2011cop}. 
More recent work~\cite{berneburg2025} considers simultaneous-move formulations, but remains limited to grid environments without obstacles.

\textit{Contributions:}
In this work, we develop a reachability-based framework for simultaneous-move PEGs on general graphs with arbitrary capture regions. 
We propose a set-based dynamic programming approach to characterize the winning and losing regions of the \pursuer{}, providing necessary and sufficient conditions under worst-case encounters. 
This characterization admits a natural set-chasing interpretation~\cite{guan2023adversarial}, leading to \emph{closed-loop} strategies that adapt to the agents' actions. 
Furthermore, for states where neither player can guarantee victory, we introduce an instantaneous matrix-game formulation to analyze the resulting \emph{gray zone}, and derive upper and lower bounds on the \pursuer{}’s winning probability. 
Overall, the proposed framework bridges geometric dominance analysis and feedback-based methods, while remaining computationally tractable and applicable to general graph environments.
The effectiveness of the proposed approach is validated through numerical simulations.

%%%%%%%%%%%%%%%%%%%%%%%%%%%%%%%%%%%%%%%%%%%%%%%%%%%%%%%%%%%%%
\section{Problem Formulation}
We study a discrete-time, two-agent pursuit--evasion game (PEG) with discrete state and action spaces.
The \pursuer{} and the \evader{} move on a shared discrete environment $\S$.
Their states at time $t$ are denoted by $x_t \in \S$ and $y_t \in \S$, respectively, and their actions by $u_t \in \U(x_t)$ and $v_t \in \V(y_t)$, where $\U(x_t)$ and $\V(y_t)$ are nonempty, state-dependent action sets.
We assume deterministic, separable dynamics:
\[
x_{t+1} = f(x_t,u_t),
\qquad
y_{t+1} = g(y_t,v_t).
\]

The \pursuer{}'s capture region is given by $\capture(x_t)\subseteq \S$, and the set of evasion states is $\target\subseteq \S$.
Capture occurs when $y_t\in\capture(x_t)$, whereas successful evasion occurs when $y_t\in\target\setminus\capture(x_t)$.
The game terminates at the first occurrence of either event.

At each nonterminal time step, both players observe the current joint state $(x_t,y_t)$ and select their actions simultaneously.
Let $\Pi$ denote the set of admissible \pursuer{} strategies $\pi: \S \!\times\! \S \!\to\! \mathcal{P}(\U(\cdot))$, where $\pi(u| x,y)$ denotes the probability of selecting action $u$ at joint state $(x,y)$.
The \evader{}'s strategy $\sigma$ and admissible strategy set $\Sigma$ are defined analogously.

We consider a \emph{game of kind}: the \pursuer{} wins if capture occurs or if the game continues indefinitely without successful evasion; the \evader{} wins if successful evasion occurs.
We say that a \pursuer{} strategy $\pi$ guarantees a win from a given initial joint state if it achieves the winning condition with \emph{probability one} against every admissible \evader{} strategy.

In summary, the game is described by the tuple
\begin{equation}
    \texttt{PEG} = \langle
    \S, \mathcal{U}, \mathcal{V}, f, g, \capture{}, \target
    \rangle,
\end{equation}
and we pursue the following research objectives:

\textbf{Objective 1:}
Characterize the dominance region from which the \pursuer{} can guarantee a win, as well as the region from which the \evader{} can guarantee  finite-time evasion.

\textbf{Objective 2:}
Construct feedback strategies that guarantee the \pursuer{}'s victory from its dominance region.

\textbf{Objective 3:}
For states outside both dominance regions, derive bounds on \pursuer{}'s optimal worst-case win rate.

%%%%%%%%%%%%%%%%%%%%%%%%%%%%%%%%%%%%%%%%%%%%%%%%%%%%%%%%%%%%%
\section{Pursuer Dominance Region}
We begin by defining the \pursuer{}'s dominance region. 
\begin{definition}
The $k$-step Pursuer's Dominance Region (\dAC{}), denoted by $\Qsetx(y,k)$, is the set of \pursuer{} states such that, given the \evader{}'s current state $y$, the \pursuer{} can guarantee (with probability one) that the \evader{} does not successfully evade within $k$ time steps.
\end{definition}
\vspace{+0.1in}

To characterize the set $\Qsetx(y,k)$, we adopt a reachability-based analysis similar in spirit to the Apollonius circle construction~\cite{dorothy2024one} in continuous pursuit--evasion.
We first define the one-step reachable sets, which contain all states an agent can reach in the next time step from its current state:
\begin{equation}
\label{eqn:reachable-set}
\begin{aligned}
\Rsetx(x_t) &= \{x_{t+1} \mid \exists u\in\U(x_t) \text{ s.t. } x_{t+1} = f(x_t,u)\}, \\
\Rsety(y_t) &= \{y_{t+1} \mid \exists v\in\V(y_t) \text{ s.t. } y_{t+1} = g(y_t,v)\}.
\end{aligned}
\end{equation}

Given a set of \pursuer{} states $\mathcal{X}$, we define the inverse reachability operator which returns the set of states that can reach $\mathcal{X}$ in one time step:
\begin{equation}
\label{eqn:inv-reachable-set}
\begin{aligned}
\Rinvsetx(\mathcal{X}) &= \{x \mid \exists u\in\U(x) \text{ such that } f(x,u) \in \mathcal{X}\},
% \Rinvsety(\mathcal{Y}) &= \{y_t \mid \exists v \text{ such that } g(y_t,v) \in \mathcal{Y}\},
\end{aligned}
\end{equation}

For notational simplicity, we further define the next-step potential capture region and the inverse capture region as 
\begin{align*}
    &\Rcap(x_t) = \bigcup_{x_{t+1} \in \Rsetx(x_t)} \capture(x_{t+1}),\\
    &\invcapture(y_t) = \{x_t \mid y_t \in \capture(x_t)\}.
\end{align*}

\subsection{Current-step Failure}
Similar to other dynamic programming approaches, we begin with the terminal condition for the \dAC{}, which consists of all states $x_t$, given the current \evader{} state $y_t$, that do not result in an immediate \pursuer{} defeat. 
%
% We denote this set by $\Qsetx(y_t, 0)$.
\begin{equation}
    \label{eqn:Q-set-0}
    \Qsetx(y_t,0) = \{x_t \mid y_t \in \capture(x_t) \cup \target^c\},
\end{equation}
where $\target^c = \S \setminus \T$ denotes the complement of the set of evasion states.
In words, a state pair $(x_t, y_t)$ does not constitute a \pursuer{}'s defeat if either the \evader{} is captured by the \pursuer{} (i.e., $y_t \in \capture(x_t)$), 
or the \evader{} has not reached an evasion state (i.e., $y_t \in \target^c$).

Since $\Qsetx(y_t,0)$ dictates the fundamental requirement for the \pursuer{} to survive the current step, any $k$-step dominance guarantee must recursively satisfy this baseline. For this reason, we denote $\Preqx(y_t) \triangleq \Qsetx(y_t,0)$ and refer to it as the \emph{required set}, serving as the foundational building block for constructing the $k$-step \dAC{} in the sequel.

\begin{figure*}[t!]
    \centering
    \includegraphics[width=\linewidth]{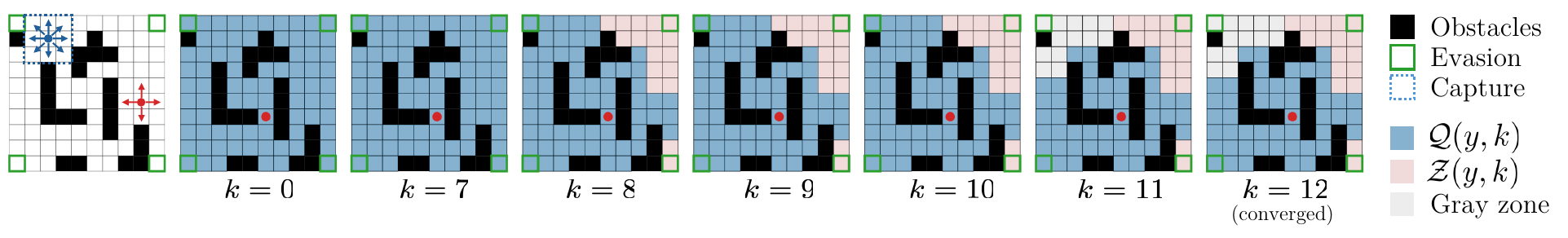}
    \vspace{-0.35in}
    \caption{
    An illustrative example of the sequence of Q-sets and Z-sets in a $10\times10$ grid environment. 
    The red dot indicates the current \evader{} position $y_t$ for which the Q-sets and Z-sets are visualized. 
    The \pursuer{} has 8-way mobility and a $3\!\times\!3$ capture region, whereas the \evader{} has 4-way mobility.
    }
    \label{fig:QZ-sequence}
    \vspace{-0.2in}
\end{figure*}

\subsection{Single-step \dAC{}}

Now, we examine the conditions under which the \pursuer{} avoids immediate defeat at the next time step.

\begin{proposition}
\label{prop:single_step_dac}
Given the current states $x_t$ and $y_t$, the \pursuer{} can guarantee that it does not lose the game at time $t+1$ if and only if $x_t \in \Qsetx(y_t,1)$, where:
\begin{align}
\label{eqn:P-req}
\Qsetx(y_t,1) =
&\; \invcapture(y_t) \\
&\midcup{} \Big(\Preqx(y_t) \midcap{} \Rinvsetx \big( 
\underbrace{\midcap{y_{t+1} \in \Rsety(y_t)} \Qsetx(y_{t+1}, 0)}_{\Qplus(y_t, 1)}
\big) \Big). \nonumber
\end{align}
\end{proposition}

\begin{proof}[Proof Sketch]

There are two cases where the \pursuer{} can deterministically prevent its defeat at $t+1$:
\begin{enumerate}
    \item[(i)] the \pursuer{} captures the \evader{} at the current time step, i.e., $x_t \in \invcapture(y_t)$;
    \item[(ii)] the \pursuer{} survives both stages $t$ and stage $t+1$.
\end{enumerate}
For the second case, survival at time $t$ requires $x_t \in \Preqx(y_t)$. 
To further ensure survival at time $t+1$, the \pursuer{} must reach a state $x_{t+1}$ such that, \emph{for all} possible states $y_{t+1} \in \Rsety(y_t)$ chosen by the \evader{}, the pair $(x_{t+1}, y_{t+1})$ does not result in \pursuer{}'s defeat, i.e., $x_{t+1} \in \Qsetx(y_{t+1}, 0)$. 
This yields the condition
$x_{t+1} \!\in\! \Qplus(y_{t}, 1)$ in~\eqref{eqn:P-req}, where $\Qplus(y_{t}, 1)$ is the set of  desirable next \pursuer{} states.
To reach one of these desirable states, the \pursuer{} must satisfy $x_t \!\in\! \Rinvsetx(\Qplus)$.
%
% A formal proof is presented in Appendix~\ref{appdx-sec:prop-proof}.
\end{proof}

\vspace{-0.15in}
\subsection{$k$-step \dAC{}}
Building upon the analysis for $1$-step \dAC{}, we can generalize the dominance region to arbitrary time horizons using a recursive set-based dynamic programming approach~\cite{shishika2023dynamic}.

\begin{restatable}{theorem}{kth} \label{thm:k_step_dac}
Let $x_t$ and $y_t$ be the current states. The \pursuer{} can guarantee that it does not lose the game within $k$ time steps if and only if $x_t \in \Qsetx(y_t, k)$, which satisfies:
\vspace{-0.05in}
\begin{align} \label{eqn:Q-set-k}
    \Qsetx(y_t,k) &=
\invcapture(y_t) \\[-0.03in]
&\hspace{-0.1in}\midcup{} \Big(\Preqx(y_t) \midcap{} \Rinvsetx \big( 
\midcap{y_{t+1} \in \Rsety(y_t)} \Qsetx(y_{t+1}, k-1)
\big) \Big). \nonumber
\end{align}
\end{restatable}

\begin{proof}
We proceed by induction on $k$.

\noindent
\underline{\textit{Base Case:}} 
Proposition~\ref{prop:single_step_dac} shows that $x_t \!\in\! \Qsetx(y_t, 1)$ is necessary and sufficient for \pursuer{} to avoid losing at time $t+1$.

\noindent
\underline{\textit{Inductive Step:}} 
Suppose the statement holds for ${k-1}$, we show that it holds for $k$.

\textit{Sufficiency:}  
If $x_t \!\in\!\! \Qsetx(y_t,k)$, then either (i) $x_t\! \in\!\! \invcapture\!(y_t)$, which ensures immediate capture, or 
(ii) $x_t \in \Preqx(y_t) \cap  \Rinvsetx \big(
\cap_{\scriptscriptstyle y_{t+1}\in\Rsety(y_t)}\Qsetx(y_{t+1}, k-1)\big)$. 
For case (ii), $x_t \!\in\! \Preqx(y_t)$ guarantees the \pursuer{}'s survival at the current time step. 
Moreover, $x_t \in \Rinvsetx \big(
\cap_{\scriptscriptstyle y_{t+1}\in\Rsety(y_t)}\Qsetx(y_{t+1}, k-1)\big)$ ensures that the \pursuer{} can select a reachable state $x_{t+1}$ such that for all $y_{t+1} \in \Rsety(y_t)$, we have $x_{t+1} \in \Qsetx(y_{t+1},k-1)$. 
By the inductive hypothesis, this ensures that the \pursuer{} can deterministically avoid defeat for the remaining $k-1$ steps, completing the proof of sufficiency.

\textit{Necessity:} 
Suppose $x_t \notin \Qsetx(y_t,k)$. Then $x_t \notin \invcapture(y_t)$, so capture does not occur at time $t$, and furthermore
\[
x_t \notin \Preqx(y_t) \midcap{}  \Rinvsetx \!\Big(
\mathop{\scalebox{1.8}{$\cap$}}\limits_{\scriptscriptstyle y_{t+1}\in\Rsety(y_t)}
\hspace{-0.15in}
\Qsetx(y_{t+1},k-1)
\Big).
\]
This implies either:
(i) $x_t \notin \Preqx(y_t)$, leading to immediate loss at time $t$; or 
(ii) for every $x_{t+1}\in\Rsetx(x_t)$, there exists some $y_{t+1}\in\Rsety(y_t)$ such that 
$x_{t+1}\notin\Qsetx(y_{t+1},k-1)$. 
By the inductive hypothesis, being in $\Qsetx(y_{t+1},k-1)$ is necessary to avoid losing over the next $k-1$ steps. 
Hence, starting from $x_t \notin \Qsetx(y_t,k)$, the \pursuer{} cannot deterministically avoid defeat within $k$ steps, therefore establishing necessity.
\end{proof}

Since $k$-step dominance necessarily implies $\ell$-step dominance for all $\ell \leq k$, we have the following remark.
\begin{remark}
    The sequence of \dAC{}s forms a non-increasing sequence of sets. Formally, for all $y \in \S$,
    \[
        \Qsetx(y, 0) \supseteq \Qsetx(y, 1) \supseteq \Qsetx(y, 2) \supseteq \cdots.
    \]
\end{remark}

\begin{theorem}
    \label{thm:Q-set-convergence}
    If the state space is finite, i.e., $|\S| < \infty$, then the \dAC{} iteration~\eqref{eqn:Q-set-k} converges in a finite number of steps.  
\end{theorem}
\begin{proof}
    This follows directly from the finiteness of the Q-sets and the fact that they form a non-increasing sequence.
\end{proof}

\subsection{Indefinite \dAC{}}
As Theorem~\ref{thm:Q-set-convergence} ensures finite-step convergence of the Q-sets, we denote the converged Q-sets as $\Qsetinf(y)$, which leads to the following corollary on indefinite survival.

\begin{corollary}
    Given the current states $x_t$ and $y_t$, 
    the \pursuer{} can guarantee a win if and only if $x_t \in \Qsetinf(y_t)$.
\end{corollary}
    
\begin{proof}
    The limiting sets satisfy the fixed-point relation:
    \[
        \vspace{-0.1in}
        \Qsetinf(y) =
        \invcapture(y) 
        \midcup{} \Big(\Preqx(y) \midcap{} \Rinvsetx \big( 
        \hspace{-0.05in}\midcap{y' \in \Rsety(y)} \hspace{-0.05in}\Qsetinf(y')
        \big) \Big).
        % \vspace{-0.05in}
    \]
    Thus, if $x_t \in \Qsetinf(y_t)$, the \pursuer{} can either capture immediately or move to $x_{t+1} \in \Qsetinf(y_{t+1})$ for all possible $y_{t+1} \in \Rsety(y_t)$.
    Since $\Qsetinf(y_t) \subseteq \Qsetx(y_t, k)$ for any $k \geq 1$, the \pursuer{} can survive for $k$ steps starting at $x_t \in \Qsetinf(y_t)$. 
    Transitioning to $x_{t+1} \in \Qsetinf(y_{t+1}) \subseteq \Qsetx(y_{t+1}, k)$ again guarantees \pursuer{}'s survival for another $k$ steps. 
    Repeating this selection guarantees either capture or indefinite prevention of evasion, both of which constitute a \pursuer{} win.

    Conversely, if $x_t\notin\Qsetinf(y_t)$, then $x_t\notin\Qsetx(y_t,K)$ for some finite $K$.
    By Theorem~\ref{thm:k_step_dac}, the \pursuer{} cannot guarantee avoidance of defeat within $K$ steps and therefore cannot guarantee a win over the infinite horizon.
\end{proof}

Fig.~\ref{fig:QZ-sequence} presents a sequence of the \dAC{} (Q-sets) generated for a specific \evader{} state marked with the red dot. 
States within the \dAC{} are shown in blue.

%%%%%%%%%%%%%%%%%%%%%%%%%%%%%%%%%%%%%%%%%%%%%%%%%%%%%%%%%%%%%
\section{Pursuer's Losing Set and Gray Zone}
Conversely, we also characterize the states from which the \pursuer{} is guaranteed to lose (i.e., the \evader{} reaches an evasion state without being captured) regardless of the \pursuer{}'s strategy.

\begin{restatable}{theorem}{zthm}
\label{thm:z-set}
The \pursuer{} is guaranteed to lose within $k$ time steps if and only if its current state $x_t \in \Zsetx(y_t, k)$, where the Z-sets are defined through recursion:
\begin{subequations}
\begin{align} 
\label{eqn:Z-set-0}
&\Zsetx(y_t, 0) = \big\{x_t \mid y_t \in  \target \setminus \capture(x_{t}) \big\} \\
 \label{eqn:Z-set-k}
&\Zsetx(y_t,k) = \Zsetx(y_t,0) \midcup{} \big( (\S \setminus \invcapture(y_t)) \;\cap \\[-0.05in]
&\hspace{+0.5in} \midcup{y_{t+1}\in \Rsety(y_t)} \hspace{-0.15in}  \{\,x_t \mid \Rsetx(x_t) \subseteq \Zsetx(y_{t+1},k-1)\} \big). \nonumber
\end{align}
\end{subequations}
\end{restatable}

\begin{proof}[Proof Sketch]
The high-level inductive argument is similar to that of Theorem~\ref{thm:k_step_dac}. 
For $k=0$, \eqref{eqn:Z-set-0} characterizes all \pursuer{} states such that the given \evader{} state $y_t$ lies in the evasion region $\target$ without being captured, which results in an immediate loss for the \pursuer{}. 

For the inductive step, the \evader{} can guarantee the \pursuer{}'s defeat within $k$ steps if and only if either the \pursuer{} has already lost at the current step (i.e., $x_t \in \Zsetx(y_t,0)$), or capture has not occurred and there exists an admissible \evader{} move $y_{t+1}$ such that every reachable \pursuer{} state lies in the corresponding $(k-1)$-step losing set (i.e., $\Rsetx(x_t) \subseteq \Zsetx(y_{t+1},k-1)$).
Taking the union over all admissible \evader{} next states, while excluding current capture states, yields~\eqref{eqn:Z-set-k}.
\end{proof}

Since losing within $k$ steps necessarily implies losing within $\ell$ steps for all $\ell \geq k$, we have the following remark.
\begin{remark}
    The sequence of \pursuer{} losing sets forms a non-decreasing sequence of sets. Formally, for all $y \in \S$,
    \[
        \Zsetx(y, 0) \subseteq \Zsetx(y, 1) \subseteq \Zsetx(y, 2) \subseteq \cdots.
    \]
\end{remark}

Similar to the \dAC{} sets, we also have the following convergence result for the losing sets.
\begin{theorem}
    \label{thm:Z-set-convergence}
    If $|\S| < \infty$, then the losing set iteration~\eqref{eqn:Z-set-k} converges in a finite number of steps. 
\end{theorem}

\begin{remark}
    We have $\Qsetx(y, k) \cap \Zsetx(y, k) = \emptyset$ for all $y \in \S$ and $k \in \mathbb{Z}_{\geq 0}$, due to the necessity and sufficiency of the Q-sets and Z-sets.
\end{remark}

An example sequence of Z-sets is also presented in Fig.~\ref{fig:QZ-sequence}. 

\subsection{Gray Zone}
\label{sec:gray-set}

The dominance region $\Qsetx(y,k)$ and the losing set $\Zsetx(y,k)$ rigorously characterize the states from which the Pursuer can either guarantee not losing or is certain to fail.
However, due to the simultaneous-move structure of the discrete pursuit–evasion game, there often exist states from which neither victory nor defeat can be guaranteed with certainty.

We formally define this probabilistic ``gray zone'' as the complement of the two $k$-step guaranteed sets:
\begin{equation}
    \Gset(y_t, k) = \mathcal{S} \setminus \big( \Qsetx(y_t, k) \cup \Zsetx(y_t, k) \big).
\end{equation}

If $x_t\in\Gset(y_t,k)$, the \pursuer{} cannot guarantee that it avoids defeat over the next $k$ steps, while the \evader{} cannot guarantee successful evasion within the $k$ steps.
The infinite-horizon gray zone is defined as
\[
\Gset_\infty(y)=\S\!\setminus\!\bigl(\Qsetinf(y)\cup \Zsetx_\infty(y)\bigr).
\]

Specifically, when the \pursuer{} is in the gray zone, both agents face a \emph{dilemma} either currently or at a future time step. 
Consider the gray region in Fig.~\ref{fig:QZ-sequence} with $k=11$ and a scenario illustrated in Fig.~\ref{fig:gray-zone-dilemma}.
At the current time step, the \pursuer{} lies in the gray zone. 
Over the next two steps, the \evader{} and \pursuer{} move toward points A and B, respectively, and the corresponding Q/Z-sets and gray region are shown. 
The \evader{} has two possible escape options: the bottom-left or bottom-right evasion state. 
To intercept successfully, the \pursuer{} must exploit its diagonal motion and follow a shortest path. 
However, the L-shaped obstacle forces the \pursuer{} to commit to a direction. 
Since the capture timing is tight, once committed, the \pursuer{} cannot adjust if the \evader{} selects the other escape route. 
Similarly, the \evader{} cannot switch without being captured. 
Consequently, the outcome is determined almost entirely by this single decision: if both agents choose the same evasion state, the \pursuer{} wins; otherwise, the \evader{} succeeds.
Since a competitive game is being played, the agents may naturally randomize their decisions, which leads to probabilistic outcomes. 

\begin{figure}[t]
    % \vspace{-0.1in}
    \centering
    \includegraphics[width=\linewidth]{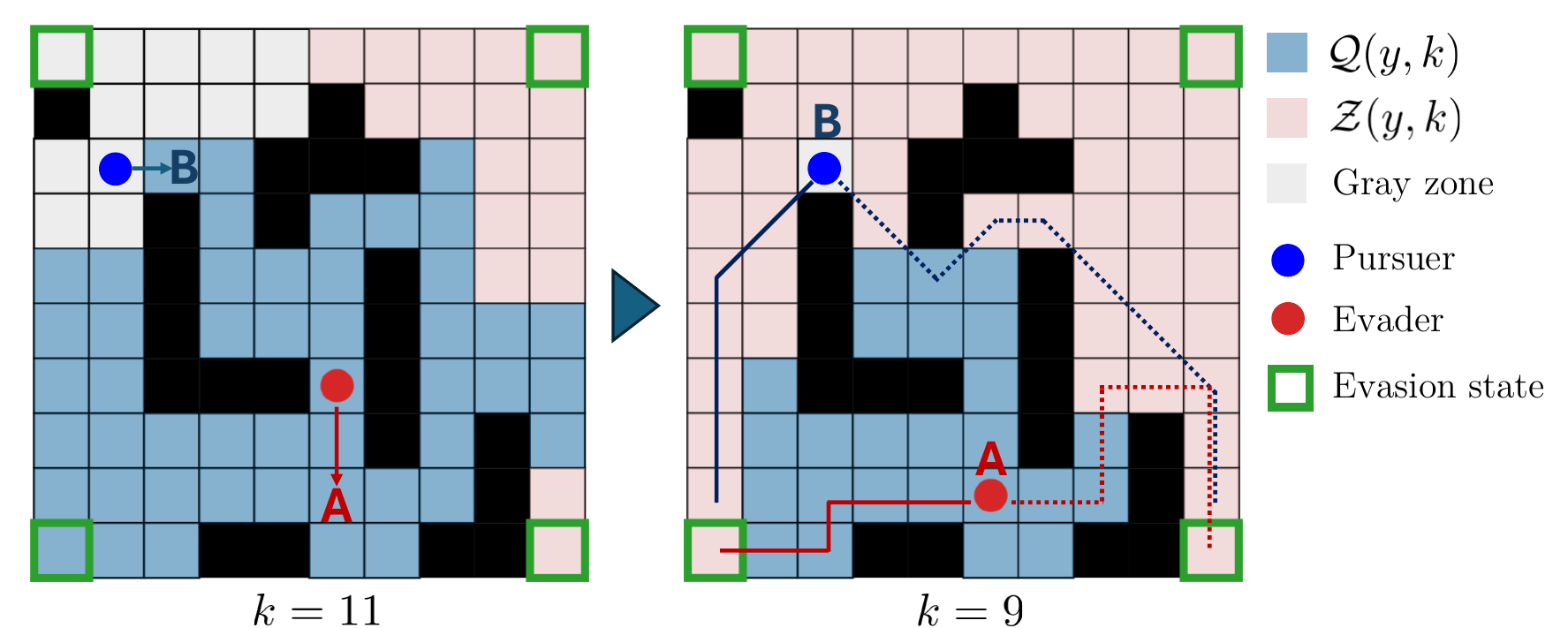}
    \vspace{-0.3in}
    \caption{An illustration of the dilemma the \pursuer{} faces in the gray zone.} 
    \label{fig:gray-zone-dilemma}
    \vspace{-0.25in}
\end{figure}

%%%%%%%%%%%%%%%%%%%%%%%%%%%%%%%%%%%%%%%%%%%%%%%%%%%%%%%%%%%%%
\section{Feedback Policies From the Dominance Sets}
In this section, we construct \emph{feedback} policies based on the Q-sets and Z-sets. 
Given the current state pair $(x_t, y_t)$, three cases arise:
\begin{enumerate}
    \item[(a)] Pursuer dominance: $x_t \in \Qsetinf(y_t)$; the \pursuer{} can capture or indefinitely prevent evasion;
    \item[(b)] Evader dominance: $x_t \in \Zsetx(y_t, k)$ for some $k$; the \evader{} is guaranteed to evade within $k$ steps;
    \item[(c)] Gray zone: $x_t \in \Gset_\infty(y_t)$; the outcome is undetermined/probabilistic.
\end{enumerate}

\noindent
We first examine the first two cases, which admit deterministic outcomes.

\subsection{Pursuer Dominance Region}
Under case (a), the \pursuer{} is guaranteed to win by capture or by preventing evasion indefinitely. 
Given the current state pair $(x_t, y_t)$, the \pursuer{} should simply select its next state
\begin{equation}
    \label{eqn:pursuer-strategy}
    x_{t+1} \in \Rsetx(x_t) \cap \midcap{y_{t+1} \in \Rsety(y_t)} \Qsetinf(y_{t+1}),
\end{equation}
so that the next \pursuer{} state lies in the corresponding Q-set, regardless of the action taken by the \evader{}. 
Such a state $x_{t+1}$ is always reachable from $x_t$ due to the reachability construction, in particular the use of the $\Rinvsetx$ operator in~\eqref{eqn:Q-set-k}.
The resulting \emph{feedback} strategy can be interpreted as ``set chasing", where the \pursuer{} continually chases a set of desirable next states given the observed \evader{} state $y_t$.

Note that the above feedback strategy does not prioritize capturing the \evader{}, even when capture can be guaranteed. 
To encode a preference for capture over indefinite evasion prevention, one may select $x_{t+1}$ in~\eqref{eqn:pursuer-strategy} so as to reduce the distance to the current state $y_t$.

Since we consider a game-of-kind formulation, both capture and no evasion constitute equivalent losing outcomes for the \evader{}. 
Therefore, under case (a), the \evader{} is guaranteed to lose regardless of its strategy. 
In practice, the \evader{} may adopt heuristic strategies such as moving towards the nearest evasion state, or towards a state in the Z-set, or simply maintaining distance from the \pursuer{}.

\subsection{Evader Dominance Region}
Case (b) guarantees that the \evader{} evades within $k$ steps. 
To minimize the time to evasion, the \evader{} first computes
\[
k_{\mathrm{min}} = \min \{k \mid x_t \in \Zsetx(y_t, k)\}.
\]
To make progress, the \evader{} targets the $(k_{\mathrm{min}}\!-\!1)$-step Z-set and selects its next state from the set
\begin{equation}
    \label{eqn:evader-strategy}
    \mathcal{Z}^+(y_t, k_{\mathrm{min}})\! 
    = \!\{y' \!\in \!\Rsety(y_t) \vert \Rsetx(x_t) \!\subseteq\! \Zsetx(y', k_{\mathrm{min}}\!-\!1)\}.
\end{equation}
By choosing any $y_{t+1} \in \mathcal{Z}^+(y_t, k_{\mathrm{min}})$, the \evader{} ensures that all possible \pursuer{} responses lie in $\Zsetx(y_{t+1}, k_{\mathrm{min}}-1)$, guaranteeing evasion within $k_{\mathrm{min}}-1$ steps.

Although the \pursuer{} is destined to lose within $k_{\mathrm{min}}$ steps, it can attempt to delay this outcome. 
In particular, the \pursuer{} should \emph{avoid} transitions to $x_{t+1} \in \Zsetx(y_{t+1}, k_{\mathrm{min}}-2)$, which would reduce the remaining survival time by two steps rather than just one. 
Since the \pursuer{} does not know the \evader{}'s choice of $y_{t+1}$, it should select $x_{t+1}$ to minimize the number of $y_{t+1} \in \mathcal{Z}^+(y_t, k_{\mathrm{min}})$ for which $x_{t+1} \in \Zsetx(y_{t+1}, k_{\mathrm{min}}-2)$.

\begin{figure*}[!t]
    \centering
    \includegraphics[width=\linewidth]{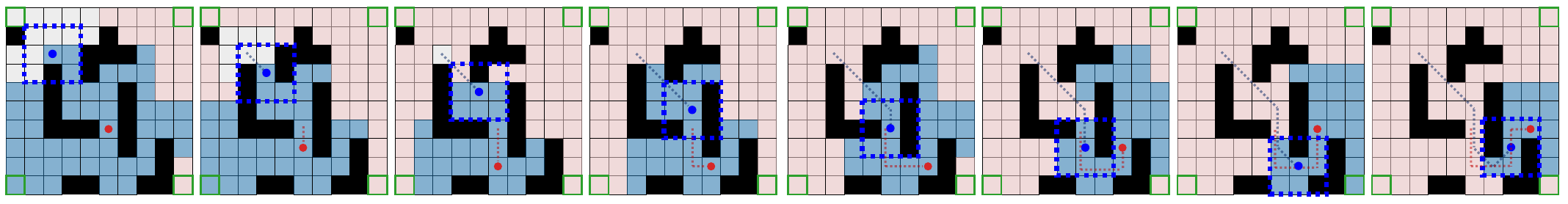}
    \vspace{-0.3in}
    \caption{\hspace{-0.1in} Game trajectories and the corresponding dominance sets at each \evader{} state. 
    The blue \pursuer{} remains within the Q-set and captures the \evader{}.}
    \label{fig:capture-sequence}
    \vspace{-0.2in}
\end{figure*}

\subsection{Gray Zone}

When $x_t \in \Gset_\infty(y_t)$, neither the \pursuer{} nor the \evader{} can guarantee their victory. 
This arises because some future joint states favor the \pursuer{}, while others favor the \evader{}, as shown in Fig.~\ref{fig:gray-zone-dilemma}.
Since the next state is jointly determined by both agents, the outcome becomes inherently uncertain.

In this case, we adopt a zero-sum matrix game formulation and allow both agents to randomize. 
Let the \pursuer{}'s feasible next states be indexed by $i$, where $x^{(i)}_{t+1} \in \Rsetx(x_t)$, and the \evader{}'s next states be indexed by $j$, where $y^{(j)}_{t+1} \in \Rsety(y_t)$. 
We construct a payoff matrix $M$, where each entry $M_{ij}$ evaluates the outcome of the next joint state $(x^{(i)}_{t+1}, y^{(j)}_{t+1})$.

Since states in the Q-set yield a winning probability of one for the \pursuer{}, while states in the Z-set yield zero, we define the payoff matrix as
\begin{equation}
    \label{eqn:payoff_matrix}
    M_{ij} = 
    \begin{cases} 
      1 & \text{if } x^{(i)}_{t+1} \in \Qsetinf\big(y^{(j)}_{t+1}\big), \\
      0 & \text{if } x^{(i)}_{t+1} \in \Zsetx\big(y^{(j)}_{t+1}, k\big)\ \text{for some } k, \\
      r & \text{otherwise},
    \end{cases}
\end{equation}
where $r \in \{0, 0.5, 1\}$ reflects different levels of optimism, as discussed below.

\begin{itemize}[leftmargin=1em]
    \item \textit{Optimistic}: $r=1$, where the \pursuer{} assumes that transitioning into the gray zone will always result in its victory;
    \item \textit{Pessimistic}: $r=0$, where the \pursuer{} assumes that entering the gray zone will always result in its defeat;
    \item \textit{Neutral}: $r=0.5$, where the \pursuer{} assumes a 50/50 chance in the gray zone.
\end{itemize}

We assume that the \pursuer{} maximizes the payoff while the \evader{} minimizes it, which leads to the following max-min optimization problem:
\begin{equation}
    \label{eqn:max-min-M}
    \max_{\pi \in \Delta_{|\Rsetx(x_t)|}}
    \min_{\sigma \in \Delta_{|\Rsety(y_t)|}}
    \bm{\pi}^\top M \bm{\sigma}, 
\end{equation}
where $\Delta_n$ denotes the $n$-dimensional simplex. 
The above max-min problem can be solved via the following linear program~\cite{owen2013game}:
\begin{alignat}{3}  
    \label{eqn:LP}
    \max_{\bm{\pi}, v} ~~ &~~ v \\
    \text{subject to} ~~ 
    & \sum_{i}[\bm{\pi}]_i M_{ij} \ge v, \quad && \forall\, j \text{ with } y^{(j)}_{t+1} \in \Rsety(y_t), \nonumber \\
    & \bm{1}^\top \bm{\pi} = 1, ~~ \bm{\pi} \ge 0. &&  \nonumber
\end{alignat}
The optimal solution $\bm{\pi}^*$ is a probability distribution over the \pursuer{}'s reachable set $\Rsetx(x_t)$, and its element $[\bm{\pi}^*]_{i}$ gives the probability of selecting $x^{(i)}_{t+1}$ as the next \pursuer{} state. 
The \evader{}'s optimal strategy can be obtained from the dual variables of~\eqref{eqn:LP}.

\subsection{Bounds on Gray-Zone Winning Probability}
To properly bound the winning probability in the gray zone, we restrict our attention to \emph{dominance-rational} strategies, where agents act optimally within the dominance sets.

\begin{definition}
A \pursuer{} strategy is \emph{dominance-rational} if it guarantees a win whenever $x_t \in \Qsetinf(y_t)$. 
Similarly, an \evader{} strategy is \emph{dominance-rational} if it guarantees the \pursuer{}'s defeat whenever $x_t \in \Zsetx(y_t, k)$ for some $k$.
\end{definition}

Dominance rationality can be viewed as a necessary condition for optimal or equilibrium strategies.

Note that the entries of matrix $M$ capture the estimated winning probability of the \pursuer{} at each possible next state. 
The optimal value $v^*$ from~\eqref{eqn:LP} therefore corresponds to the equilibrium winning probability of the \pursuer{}, assuming that the estimates in $M$ are accurate.
While the entries $1$ and $0$, corresponding to the Q-sets and Z-sets, are exact under dominance-rational strategies, the entries involving $r$ approximate the outcome in the gray zone. 
Nevertheless, the true winning probability always lies in $[0,1]$. 

Let $M_{\mathrm{true}}$ denote the accurate estimate under an arbitrary pair of dominance-rational strategies. We have the following element-wise inequality
\[
    M_{\pess} \leq M_{\true} \leq M_{\optism},
\]
where $M_{\pess}$ corresponds to the pessimistic case with $r=0$ and $M_{\optism}$ corresponds to the optimistic case with $r=1$. 

\begin{lemma}
\label{lemm:value-inequality}
Let $v^*_{\pess}$, $v^*_{\true}$, and $v^*_{\optism}$ be the optimal values of~\eqref{eqn:LP} corresponding to $M_{\pess}$, $M_{\true}$, and $M_{\optism}$, respectively. 
Then, we have
\[
    v^*_{\pess} \leq v^*_{\true} \leq v^*_{\optism}.
\]
\end{lemma}

\begin{proof}
We show that if $A, B\!\in\! \mathbb{R}^{n\times m}$ satisfy $A\!\geq\! B$ elementwise, then we have $v^*_A \geq v^*_B$, where
\begin{equation*}
    v^*_A = \max_{\pi \in \Delta_n} \min_{\sigma \in \Delta_m} \pi^\top A \sigma, 
    \quad
    v^*_B = \max_{\pi \in \Delta_n} \min_{\sigma \in \Delta_m} \pi^\top B \sigma.
\end{equation*}
Since $A \geq B$ elementwise and $\pi,\sigma \geq 0$ are probability vectors, we have
\(
    \pi^\top A \sigma \geq \pi^\top B \sigma
\)
for all $(\pi,\sigma) \in \Delta_n \times \Delta_m$.

Let $\pi_B^* \in \Delta_n$ be an optimal row policy for game $B$, so that
\(
    \min_{\sigma \in \Delta_m} (\pi_B^*)^\top B \sigma = v_B^*.
\)
Since $v_A^*$ is the maximum worst-case payoff over all $\pi \in \Delta_n$, evaluating game $A$ at the specific policy $\pi_B^*$ yields
\begin{align*}
    v_A^*
    &= \max_{\pi \in \Delta_n} \min_{\sigma \in \Delta_m} \pi^\top A \sigma
    \geq \min_{\sigma \in \Delta_m} (\pi_B^*)^\top A \sigma \\
    &\geq \min_{\sigma \in \Delta_m} (\pi_B^*)^\top B \sigma 
    = v_B^*.
\end{align*}
Therefore, $v_A^* \geq v_B^*$.
\end{proof}

Define the \pursuer{}'s optimal worst-case winning probability as
\[
 V(x_t,y_t) = \sup_{\pi \in \Pi} \inf_{\sigma \in \Sigma} \mathbb{P}_{\pi, \sigma}\big( \text{Pursuer wins} \mid x_t, y_t\big).
\]

Since the winning probability in the gray zone lies in $[0,1]$, which corresponds to the $M_\pess$ and $M_\optism$ matrices, we have the following bounds.

\begin{corollary}
\label{cor:bounds}
The pessimistic and optimistic matrix games provide lower and upper bounds on $V(x_t,y_t)$, respectively. 
Specifically, at any joint state $(x_t,y_t)$ in the gray zone,
\[
v_{\mathrm{pess}}^\ast
\le V(x_t,y_t)
\le v_{\mathrm{optsm}}^\ast.
\]
\end{corollary}

%%%%%%%%%%%%%%%%%%%%%%%%%%%%%%%%%%%%%%%%%%%%%%%%%%%%%%%%%%%%%

\section{Numerical Experiments}
For clarity of exposition, we present results on the grid environment shown in Fig.~\ref{fig:QZ-sequence}. 
The \pursuer{} has 8-way mobility and a $3\!\times\!3$ capture region, whereas the \evader{} has 4-way mobility.
The proposed dominance set analysis and the set-based dynamic program, however, apply to arbitrary graphs and irregular \pursuer{} capture regions.

All results are evaluated with 1,000 trials with a maximum horizon of $50$ steps; if the \evader{} fails to reach an evasion state within this horizon, the \pursuer{} is declared the winner.

\subsection{Dominance Regions}

We first revisit the dilemma example introduced earlier in Fig.~\ref{fig:dilemma}. 
Fig.~\ref{fig:state-pos}a shows the corresponding Q-set and Z-set, illustrating that the proposed set-based dynamic programming approach correctly identifies the given \pursuer{} state (blue point) as a losing state. 

\begin{figure}[b]
    \vspace{-0.1in}
    \centering
    \includegraphics[width=\linewidth]{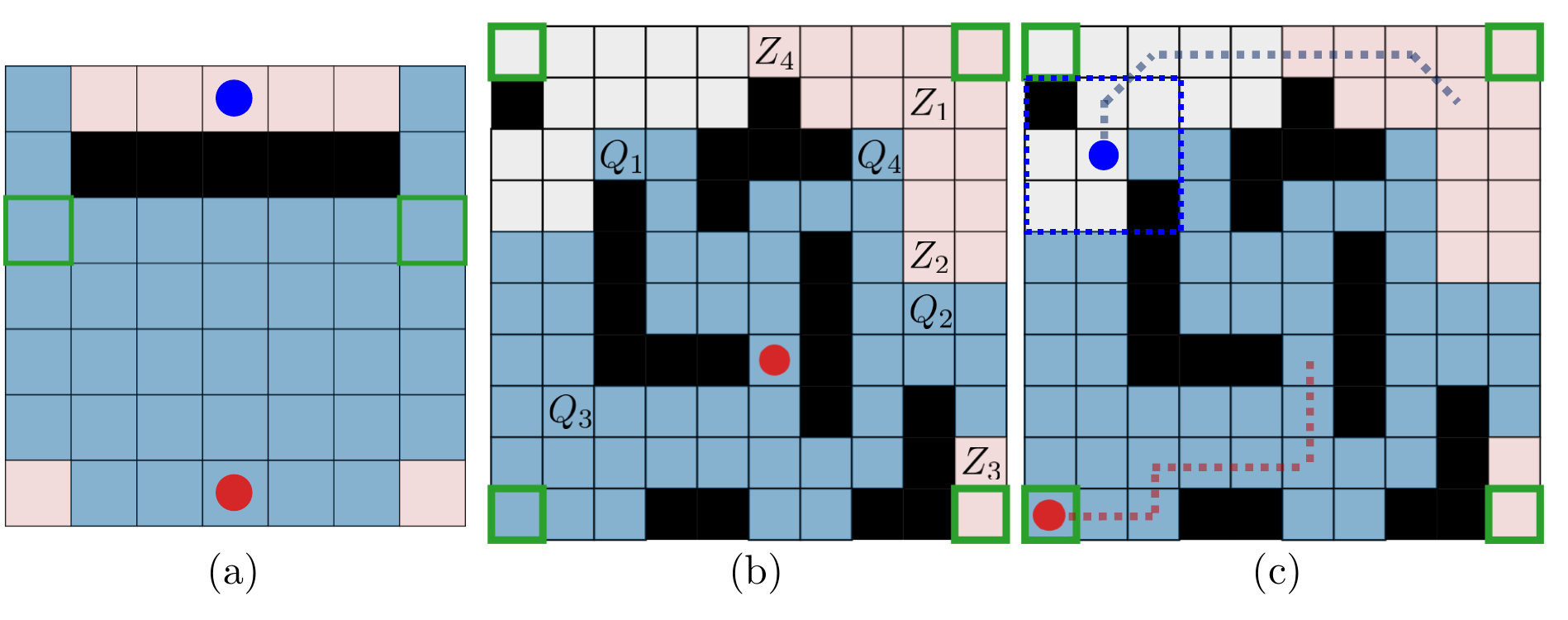}
    \vspace{-0.3in}
    \caption{
    {(a)} Visualization of Q/Z-sets for the dilemma example. 
    {(b)} Pursuer states used for performance evaluations. 
    {(c)} Game trajectory with the \pursuer{} starting from Z-set.}
    \label{fig:state-pos}
    \vspace{-0.08in}
\end{figure}

Fig.~\ref{fig:capture-sequence} illustrates the trajectories of the \pursuer{} and the \evader{}, with the \pursuer{} initially located in the Q-set (its dominance region), where its victory is guaranteed. 
In this scenario, we let the \evader{} follow a shortest path toward the bottom-right evasion state. 
The corresponding dominance sets at each \evader{} state are shown, illustrating that the \pursuer{} ``chases" the Q-sets, remains within the Q-sets at each time step, and ultimately captures the \evader{}.

Next, we compare the proposed policies derived from the dominance sets with Nash equilibrium (NE) policies. 
The NE policies are computed by treating the pursuit-evasion game as a stochastic game with a reward structure that assigns $1$ for capture, $-1$ for successful evasion, and $0$ otherwise. 
Under this zero-sum formulation, the \pursuer{} maximizes this reward while the \evader{} minimizes it, 
and we compute an NE via value iteration~\cite{guan2022hierarchical}.

Importantly, NE policies are dominance-rational, i.e., whenever a winning opportunity exists, the agents exploit it optimally. 
As such, they provide a strong benchmark, acting optimally within the dominance regions and achieving the best performance against a fully rational opponent under worst-case assumptions%
\footnote{
    The NE solution concept is inherently a game-of-degree formulation and thus differs from our game-of-kind formulation; nevertheless, it remains a natural and strong candidate for benchmarking.
}.
This makes NE particularly suitable for validating the dominance-region analysis and for estimating the true winning probabilities in the gray zone.

We evaluate the performance of different strategies at eight initial \pursuer{}  states, while the \evader{} is fixed at the red dot as shown in Fig.~\ref{fig:state-pos}b. 
Specifically, states $Q1$--$Q4$ lie in the Q-set, where an optimal \pursuer{} is expected to achieve a $100\%$ win rate. 
Similarly, states $Z1$--$Z4$ lie in the Z-set, where the \evader{} is guaranteed to evade. 

\begin{table}[t]
\vspace{+0.05in}
\centering
\caption{Simulated Pursuer Win Rate in Dominance Regions}
\vspace{-0.05in}
\setlength{\tabcolsep}{3.5pt}
\begin{tabular}{cc cc}
\toprule
\multicolumn{2}{c}{\textbf{$Q$ States}} & \multicolumn{2}{c}{\textbf{$Z$ States}} \\
\cmidrule(lr){1-2} \cmidrule(lr){3-4}
State & \pursuer{} Win Rate
& State & \pursuer{} Win Rate \\
\midrule
    $Q1$ & $1.0$ & $Z1$ & $0.0$  \\
    $Q2$ & $1.0$ & $Z2$ & $0.0$  \\
    $Q3$ & $1.0$ & $Z3$ & $0.0$  \\
    $Q4$ & $1.0$ & $Z4$ & $0.0$  \\
\bottomrule
\end{tabular}
\label{table:QZ-results}
\vspace{-0.15in}
\end{table}

Table~\ref{table:QZ-results} reports the \pursuer{} winning probabilities for the $Q$- and $Z$-states. 
For $Q$-states, we evaluate the Q-set-based \pursuer{} strategy against a Nash \evader{} strategy to avoid bias toward a specific opponent. 
Similarly, we evaluate the Z-set-based \evader{} strategy against a Nash \pursuer{} strategy for the $Z$-states.
The simulation results are fully consistent with the dominance-region analysis. 
In particular, $x_t \in \Qsetinf(y_t)$ guarantees the \pursuer{}'s victory, whereas $x_t \in \Zsetx(y_t, k)$ guarantees evasion and thus the \pursuer{}'s loss.

Next, we examine the \pursuer{} state $Z1$ marked in Fig.~\ref{fig:state-pos}b. 
The $k\!=\!8$ panel of  Fig.~\ref{fig:QZ-sequence} shows that $Z1 \in \Zsetx(y_t, 8)$, implying guaranteed evasion within eight steps. 
To validate this, we plot a sample trajectory of a Nash \pursuer{} against a Z-set-based \evader{} in Fig.~\ref{fig:state-pos}c. 
Consistent with the theory, the \evader{} successfully evades within 8 steps, even against the optimal Nash \pursuer{}.

\subsection{Gray Zone}
For the gray zone, we examine the scenario in Fig.~\ref{fig:gray-zone-dilemma} and analyze the LP solution from~\eqref{eqn:LP}. 
There are 8 feasible \pursuer{} next states and 5 feasible \evader{} next states, as illustrated in Fig.~\ref{fig:lp-construct}. 
Based on~\eqref{eqn:payoff_matrix}, the resulting payoff matrix is
\vspace{-0.1in}
\[
    M^\top = \scalebox{0.85}{$\left[
        \begin{array}{cccccccc}
            r & 0 & 0 & 0 & 0 & 0 & 0 & 0  
            \\
            0 & 0 & 0 & 0 & 0 & 0 & r & 0  
            \\
            r & r & r & r & r & r & r & 1  
            \\
            0 & 0 & 0 & 0 & 0 & 0 & 0 & 1  
            \\
            r & r & r & r & r & r & r & 1  
            \\
        \end{array}
    \right]$}.
\]

For the optimistic case ($r=1$), the LP in~\eqref{eqn:LP} yields $v^*_{\optism} \!=\! \tfrac{1}{3}$, providing an upper bound on the \pursuer{}'s winning probability. 
The resulting optimal \pursuer{} policy is 
\(
    [\tfrac{1}{3}, 0, 0, 0, 0, 0, \tfrac{1}{3}, \tfrac{1}{3}]
\),
so that the \pursuer{} randomizes uniformly among three actions: staying at its current position $x^{(0)}$, moving to $x^{(6)}$, and moving to $x^{(7)}$.
The \evader{} randomization is given by \( [\tfrac{1}{3}, \tfrac{1}{3}, 0, \tfrac{1}{3}, 0] \),
that is, the \evader{} uniformly randomizes among staying at its current position $y^{(0)}$, moving leftward to $y^{(1)}$ and rightward to $y^{(3)}$.

In contrast, for the pessimistic case ($r=0$), we obtain $v^*_{\pess} \!=\! 0$, and an optimal solution is for both agents to remain at their current positions. 
This highlights that the optimistic \pursuer{} adopts a more proactive strategy to chase the \evader{}, whereas the pessimistic \pursuer{} remains conservative.

Finally, we simulate 1,000 trajectories using the Nash policies and obtain a \pursuer{} win rate of $0.102 \in [0,\hspace{+0.01in} \tfrac{1}{3}] = [v^*_\pess, v^*_\optism]$, consistent with the bounds established in Corollary~\ref{cor:bounds}.

\begin{figure}[t]
    \centering
    \includegraphics[width=\linewidth]{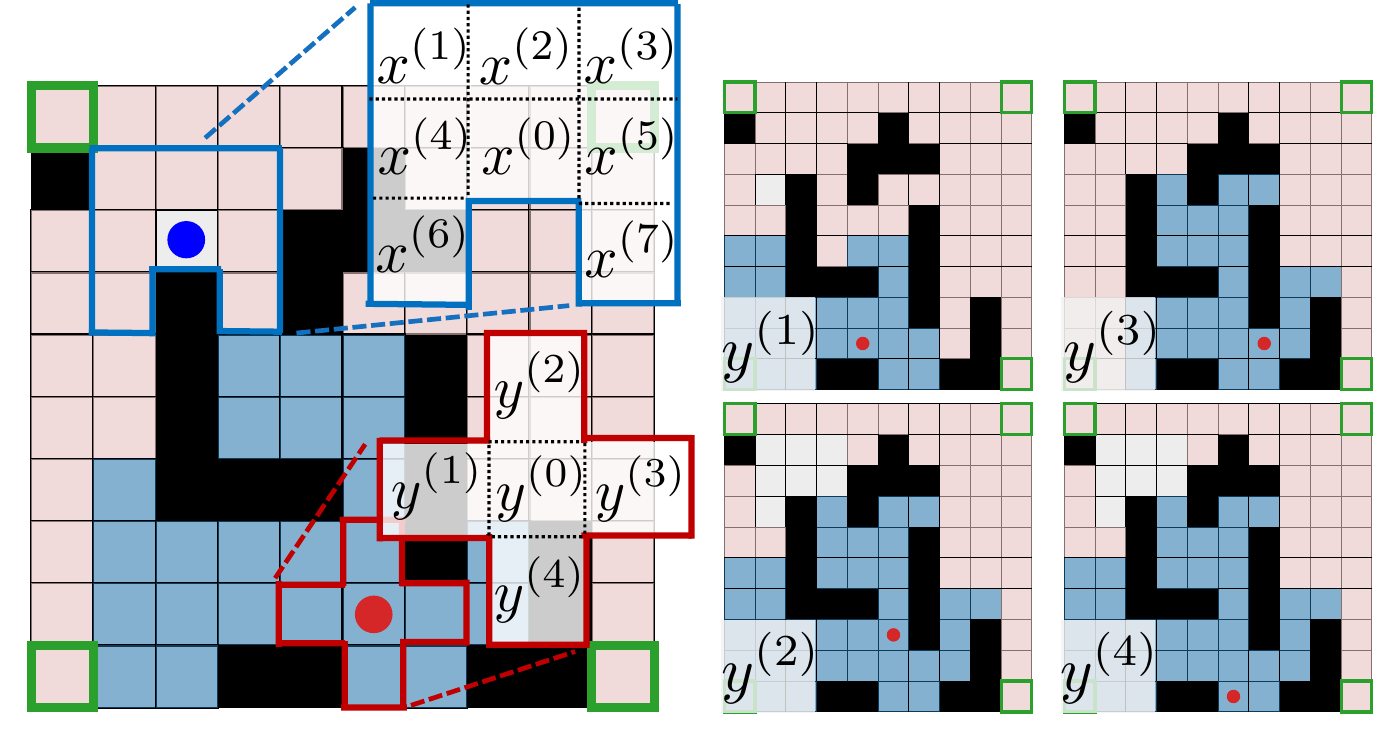}
    \vspace{-0.3in}
    \caption{Illustrative example of LP construction in the gray zone.}
    \label{fig:lp-construct}
    \vspace{-0.2in}
\end{figure}

\section{Conclusion}

In this work, we developed a reachability-based framework for analyzing dominance regions in simultaneous-move pursuit--evasion games on general graphs. 
Using a set-based dynamic programming approach, we characterized the \pursuer{}’s winning and losing regions, which lead to closed-loop feedback strategies that overcome limitations of existing open-loop analysis of dominance regions. 
For gray-zone states, where neither player can guarantee victory, we introduced a matrix-game formulation that provides bounds on the \pursuer{}’s winning probability.
We demonstrate the correctness of the dominance-region characterization and the validity of the proposed gray-zone bounds through numerical simulations.
Future work includes extending the framework to multi-agent settings and tightening gray-zone bounds.

\section{Acknowledgment}
This work has been supported in part by ARL under DCIST CRA W911NF-17-2-0181 and by ONR award N00014-23-1-2308.

\bibliographystyle{IEEEtran}
\bibliography{refs}

\end{document}